\documentclass[journal]{IEEEtran}

\usepackage{graphicx}
\usepackage[caption=false,font=footnotesize]{subfig}
\usepackage{algorithmic}
\usepackage{amsthm}
\usepackage{amsmath}
\newtheorem{theorem}{Theorem}
\newtheorem*{lemma*}{lemma}
\newtheorem{proposition}[theorem]{Proposition}

\usepackage{enumerate}
\usepackage{wrapfig}

\begin{document}
\title{An Incentive-Compatible Scheme for Electricity Cooperatives: An Axiomatic Approach}
\author{Abbas~Ehsanfar,~\IEEEmembership{Student Member,~IEEE,}
        Babak~Heydari,~\IEEEmembership{Member,~IEEE}
 
\thanks{A. Ehsanfar is with the School
of Systems and Enterprises, Stevens Institute of Technology, Hoboken,
NJ, 07030 USA (e-mail: aehsanfa@stevens.edu).}
\thanks{B. Heydari is with the School
of Systems and Enterprises, Stevens Institute of Technology, Hoboken,
NJ, 07030 USA (e-mail: babak.heydari@stevens.edu).}}

\markboth{IEEE transactions on Smart Grids}
{Ehsanfar \MakeLowercase{\textit{et al.}}: IEEE transactions on Smart Grids}

\maketitle

\begin{abstract}

This paper introduces a new scheme for autonomous electricity cooperatives, called predictive cooperative (PCP), which aggregates commercial and residential electricity consumers and participates in the electricity market on behalf of its members.  An axiomatic approach is proposed to calculate the day-ahead bid and to disaggregate the collective cost among participating consumers. The resulting formulation is shown to keep the members incentivized to both participate in the cooperative and remain truthful in reporting their expected loads. The scheme is implemented using PJM (world's largest  wholesale electricity market) real-time and day-ahead price data for 2015 and a collection of residential and commercial load profiles. The model performance of this framework is compared to that of \textit{real-time pricing} (RTP) scheme, in which wholesale market prices are directly applied to individual consumers. The results show truthful load announcement by consumers, reduction in electricity price variation for all consumers, and comparative benefits for participants.  
\end{abstract}


\IEEEpeerreviewmaketitle

\section{Introduction}
 \IEEEPARstart{T}{he} new paradigm of the electricity industry commonly known as smart grid includes active demand response (DR), distributed energy resources (DER), enhanced asset utilization, and consumer choice~\cite{rahimi2010demand}.

 The integration of demand-side management (DSM) in the future smart grid has been widely discussed in the literature, e.g.,~\cite{palensky2011demand,albadi2008summary}. In general, \textit{demand response} approaches can be divided into three categories: autonomous demand response and scheduling~\cite{parvania2010demand,gatsis2011cooperative}; price-based demand-side management~\cite{su2009quantifying}; and a mixed approach that combines these two schemes~\cite{chen2012real,mohsenian2010autonomous} in order to provide the end-consumer with simultaneous economic benefit and convenience in participation.

Price-based demand response applies variable pricing to incentivize consumers to flatten their demand profile. Several pricing schemes have been suggested in the literature including time of use (TOU), critical peak pricing (CPP), and real-time pricing (RTP)~\cite{albadi2008summary}. RTP schemes charge consumers hourly prices that reflect the real cost of electricity in the wholesale market~\cite{borenstein2002dynamic,mohsenian2010optimal,allcott2009real}. Variations of this scheme have already been introduced to commercial and residential consumers in some regions such as ComEd's residential real-time pricing (RRTP) in the state of Illinois. Nevertheless, for residential consumers, RTP scheme has not been sufficiently effective for demand response because of a number of economic and behavioral issues, most notably low elasticity of residential demand (especially at high prices) and difficulty of real-time monitoring and inconvenience that residential consumers face in prediction of prices \cite{gyamfi2013residential,allcott2011rethinking}.
 
\textit{Automated enabling systems} provide dynamic, active load management that can respond to unforeseen real-time variations in the market~\cite{spees2007demand}. Several benefits of intelligent autonomous units in \textit{smart grid} are discussed in details in~\cite{davito2010smart,peters2013reinforcement,ketter2013power}. In the spirit of these benefits, some authors have suggested combined use of real-time pricing with an automated framework in DR application: Conejo et al. and Parvania et al. address DR with RTP and propose a solution for the subsequent appliance scheduling optimization problem~\cite{conejo2010real,parvania2010demand} ; Mohsenian et al. developed an RTP residential load control for maximizing the consumer utility~\cite{mohsenian2010optimal}; and Lujano et al. propose a load management strategy that optimizes negotiation between consumers and retailers with RTP~\cite{lujano2012optimum}. In practice, automated demand response relies on home energy management systems (HEMS) to communicate with retailers, and autonomous units of other consumers~\cite{hubert2012modeling}. 

The majority of available RTP based models are developed by assuming active participation of end-consumers in the wholesale electricity market, an assumption that is unrealistic from regulatory perspective in most regions~\cite{rahimi2010demand}. Moreover, as has been argued~\cite{mohsenian2010optimal,palizban2014microgrids,santos2015multi}, a more effective practice in demand-side management is to manage aggregate behavior and incentivize demand response with price signals in a multi-agent system, something that is not directly implementable in a pure RTP scheme because real-time prices indiscriminately apply to all RTP consumers~\cite{rious2015electricity}. In principle, load aggregation benefits consumers by offsetting the individual load volatilities, thus decreasing relative load variations and enhancing the accuracy of load prediction. However, in order to benefit, consumers need to participate based on an internal contractual scheme that clearly determines how aggregate benefits or potential aggregate penalties shall be distributed among participating consumers according to their level of participation, truthfulness, load estimation, and external factors such as market prices. Such \textit{cooperative} schemes have been previously introduced ~\cite{akasiadis2013agent,veit2013multiagent,gatsis2011cooperative}, however they need to be formulated in such a way to incentivize participation and truthful information sharing and disincentive free riding, gaming, and detrimental opportunistic behavior. 

In this spirit, the goal of this paper is to introduce a scheme that flattens real-time price variations for participating consumers in an electricity \textit{cooperative} and incentivizes them to truthfully reveal their expected load information. This scheme relies on an autonomous agent that interacts with consumers and participates in the electricity market on their behalf.

The proposed scheme, called \textit{predictive cooperative} (PCP), includes four key components: load information sharing by consumers, load aggregation and forecasting, day-ahead bidding, and payment disaggregation. For each consumer, we define two key parameters: an announced load as the expected consumption and a confidence factor as a value between 0 and 1 that indicates, for the purpose of market bidding, how much the consumer prefers to rely on PCP's suggested bid, as opposed to her own intended consumption. As the first step, the PCP agent receives the load and \textit{confidence factor} for day-ahead loads from each HEMS unit. Then, PCP forecasts the short-term aggregate load using \textit{double-seasonal autoregressive smoothing method}~\cite{taylor2003short}, which we have modified to accommodate dynamic parameter updating. The PCP then calculates the individual day-ahead bid using the announced loads, confidence levels, and aggregate forecast. Finally, the total cost will be disaggregated among consumers according to an incentive-compatible formulation, developed axiomatically. 

There are two distinct advantages in combining PCP demand forecasts with individual consumers' estimates: First, the sum of values provided by consumers is not necessarily the accurate aggregate estimation and aggregate forecasting can improve the estimation accuracy~\cite{srinivasan2008energy}. In addition, consumers may be untruthful or biased in reporting their estimations for energy usage~\cite{gedra1993markets}, if market bids depend only on their reported estimates. 

Our paper makes the following contributions to the literature: First, it offers a novel incentive-compatible scheme for electricity \textit{cooperative}s to achieve a collective goal of reducing price risk for the participants while enabling the prospect of reduced average price for individual consumers, using a bottom-up, axiomatic approach, which to the best of authors knowledge have not been previously used in the literature of smart grids. The paper furthermore mathematically proves that the provided scheme is not only beneficial for individual consumers to participate, but also provides incentives to participants to adopt a truthful behavior when interacting with the scheme. The combination of these two ensures a bilateral trust between the \textit{cooperative} scheme and individual consumers. Finally, the paper uses an agent-based simulation, using real-world data from the PJM\footnote{PJM electricity wholesale market coordinates and delivers power to whole or part of 13 U.S. states plus the District of Columbia.} market to demonstrate the benefits of the proposed scheme. As a part of the simulation process, and as a side contribution, the paper also improves the accuracy of one of the standard short-term load forecasting models by migrating toward dynamic parameter updating. Although formulated for the PJM electricity market, the provided scheme and the axiomatic method used in its development can provide useful insights for formulating other cooperatives of end-consumers in commodity markets with forward/future contracts, such as agricultural water, natural gas and crops.

In the rest of this paper, we introduce the proposed framework and explain its various components. We first provide an overview of the architecture of the framework in Section~\ref{sec:framework}, then, in Section~\ref{sec:ax_model} we introduce an axiomatic approach for effective day-ahead bids and disaggregation of the overall payment into individual shares, followed by the load forecasting model in Section~\ref{sec:forecasting}. Finally, Section~\ref{sec:agentbasedsim} presentes the simulation of  the proposed scheme and its performance, using PJM market data. 

\section{Interactive Electricity Cooperatives} 

\label{sec:framework}
In this section, we introduce the concept of  \textit{predictive cooperative} (PCP), a non-profit scheme that participates in the wholesale market on behalf of its members on the one hand, and acts as a contractual scheme for their participation on the other hand. Residential and commercial electricity consumers subscribe to the PCP, which provides consumers with the day-ahead price, receives consumer information, forecasts aggregated load, places an aggregate bid in the market, and finally disaggregates the collective cost among consumers. An overview of the proposed procedure is illustrated in Fig.~\ref{fig:PCP} and the following algorithm:

\noindent\rule{\columnwidth}{0.4pt}
\textbf{Algorithm I}: PCP Daily Procedure.\\[-5pt]
\noindent\rule{\columnwidth}{0.4pt}\\[-10pt]
{\fontsize{8}{10}\selectfont
 \begin{algorithmic}[1]
\STATE \textbf{I. 11am-12pm:} Consumers provide their day-ahead expected load and confidence factor.
\STATE \textbf{II. 12pm:} PCP forecasts load for the next-day time interval (12:00am-11:59pm) and places bid in the day-ahead market.
\STATE \textbf{III. 12am:} 
\FOR{each time step for 24 hours}
  \STATE Participants consume their real-time load
  \STATE PCP receives real-time prices from real-time market
  \STATE PCP clears its account in real-time market
  \STATE PCP charge the consumers with the disaggregated hourly payment
  \IF{11am}
  \STATE Repeat from step I for the next-day bids
  \ENDIF
  \ENDFOR 
  \end{algorithmic}}
  \noindent\rule{\columnwidth}{0.4pt}
Besides an accurate load prediction algorithm, as explained in Section~\ref{sec:forecasting}, the success of the PCP agent depends on two critical components: First, since the load predicted by the PCP agent is in general different from the sum of declared day-ahead loads by individual consumers, the agent needs to combine both of these signals in order to calculate the effective day-ahead bid for each consumer. The relative weight of each signal for a given consumer depends on the level of confidence that the consumer has announced in the PCP estimation versus her own. To combine these signals as a function of the confidence factor, we propose an axiomatic approach as explained in Section~\ref{sec:dayaheadbid}.

The second issue arises because, even with a good effective day-ahead bid model, real-time consumption have discrepancy with respect to the effective day-ahead bid which translates into potential costs or surpluses for the PCP and consequently for individual consumers. Disaggregating these costs and surpluses is a function of several parameters, such as day-ahead bids, real-time consumption and market prices. Here again, we use an axiomatic approach, as explained in Section~\ref{sec:paydisag} to disaggregate the total payment for participating consumers as a function of the contributing parameters.

\begin{figure}[t]
\centering
 \includegraphics[width=0.47\textwidth]{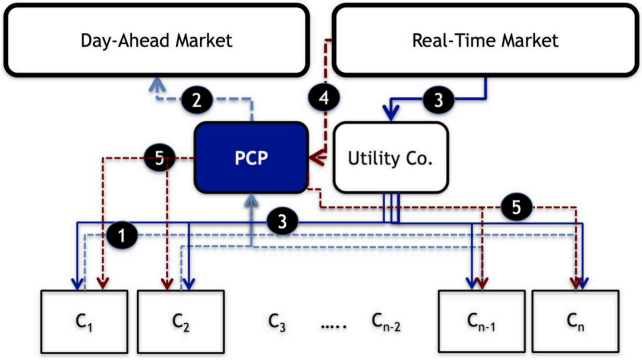}
\caption{The proposed PCP steps are: 1) PCP receives the day-ahead load and confidence factor; 2) PCP places bid in day-ahead market; 3) participants consume real-time load; 4)  PCP receives real-time cost; and 5) PCP disaggregates the cost among participants. The day-ahead model (Section. (\ref{sec:dayaheadbid})) is used for day-ahead market and the payment model (Section.(\ref{sec:paydisag})) is used for real-time market.}
\label{fig:PCP}
\end{figure}

\section{Axiomatic Day-Ahead and Payment Disaggregation Models}
\label{sec:ax_model}
In this section we propose two models, one for the day-ahead bid calculation and the other for final payment disaggregation among participants in an electricity \textit{cooperative}. In the first model, we use an axiomatic approach to formulate how the PCP should calculate day-ahead bids by combining its forecasted load - using a central forecasting algorithm as will be explained in Section \ref{sec:forecasting}-with the sum of expected loads, announced by individual consumers. The second model is used to disaggregate the total payment among participants based on a set of axioms inferred from mathematical constraints (fixed-sum etc.) and fairness in payment model. 

\subsection{Day-Ahead Model}
\label{sec:dayaheadbid}

The PCP agent uses two distinct signals to calculate its day-ahead bid: its own forecasted load ($L_{f(t)}$), and the sum of announced day-ahead loads reported by the consumers to the PCP ($L_{a(t)}=\sum_{i}{l_{a(t)}^i}$). We introduce \textit{total effective day-ahead bid} ($L_e$) as the weighted combination of these two signals. But how should the PCP agent decide about their relative weights? To formulate this decision, we further assume that each consumer has a certain level of confidence ($\rho_i$) to the PCP agent. We model this confidence ($\rho_i$) as a number between 0 (no confidence) and 1 (full confidence). Thus, the PCP needs to take into account the confidence factors of different consumers in order to combine the two signals that make the effective day-ahead bid. To facilitate the calculations, we introduce \textit{individual effective day-ahead bids}($l_{e}^i$) as a set of instrumental variables to capture individual consumers' shares of PCP day-ahead bid, so that $L_e=\sum_{i}{l_e^i}$. To formulate the value of $l_{e}^i$s, we take on an axiomatic approach by assuming the following axioms: 
\subsubsection{Boundedness}: The total effective day-ahead bid is bounded by the PCP forecasted load and the sum of individuals' declared loads:
 \begin{equation}
\exists \alpha \in [0,1]: L_{e(t)}=\alpha L_{f(t)}+(1-\alpha)L_{a(t)}
\label{form:boundary}
\end{equation}
where $L_{e(t)}$ is the aggregate day-ahead bid.

\subsubsection{Similarity in deviation direction}: For all the day-ahead loads, we assume that if the total effective bid is more(less) than the sum of reported loads, then effective individual bid is also more (less) than the individual reported load.This makes sense since, depending on the confidence factor of a given consumer, her effective bid is always between her reported bid and a number calculated by the PCP, and the direction of this deviation depends on the aggregate deviation. In other words deviation of effective bids shall have the same sign as the aggregate deviation: $sign(l_{e(t)}^i-l_{a(t)}^i)=sign(L_{e(t)}-L_{a(t)}) \nonumber$.

\subsubsection{Proportionality}: All other factors equal, effective bid deviation of a consumer is an increasing function with respect to her confidence factor: 
${\partial (l_{e(t)}^i-l_{a(t)}^i)}/{\partial {\rho}_{i(t)}}>0.$
We now need an equation for day-ahead bid that satisfies all these axioms. The formulation is not unique, yet we pick the following equation, which is simple and, as we will show, satisfies all three axioms: 
 \begin{equation}
l_{e(t)}^i=l_{a(t)}^i+{\rho}_i^2*l_{a(t)}^i*(L_{f(t)}-L_{a(t)})/\sum_{j=1}^{N}{{\rho}_jl_{a(t)}^j}.
\label{form:risksharing}
\end{equation}
And the total day-ahead bid is: $L_{e(t)}=\sum_{i}{l_{e(t)}^i}$.

\begin{proposition}
Equation~(\ref{form:risksharing}) satisfies Axioms A.1- A.3. 
\end{proposition}
\begin{proof}
First, the $L_{e(t)}$ will remain in the reasonable boundary defined by $L_{f(t)}$ and $L_{a(t)}$. According to (\ref{form:boundary}) we can find $\alpha$ equal to: $\alpha={\rho}_i^2*l_{a(t)}^i/\sum_{j=1}^{N}{{\rho}_jl_{a(t)}^j}$ which is in $[0,1]$ because $\rho_i \in [0,1]$. Second, \(l_{e(t)}^i-l_{a(t)}^i\) and \(L_{e(t)}-L_{a(t)}\) have identical signs because ${\rho}_i^2*l_{a(t)}^i$ and $\sum_{j=1}^{N}{{\rho}_jl_{a(t)}^j}$ are both positive. 
Finally, the formulation ensures positive derivatives for effective load deviation with respect to confidence level ($\rho_i$): 
\setlength{\arraycolsep}{0.0em}
\begin{eqnarray}
\frac{\partial (l_e^i-l_{a(t)}^i)}{\partial {\rho}_i}&=\frac{\rho_il_{a(t)}^i*(2\sum_{j=1}^{N}{\rho_il_{a(t)}^j}-\rho_il_{a(t)}^i)}{(\sum_{j=1}^{N}{\rho_il_{a(t)}^j})^2}\nonumber\\
&=\frac{\rho_il_{a(t)}^i(2\sum_{j=\{1,..,N\}/i}{\rho_il_{a(t)}^j}+\rho_il_{a(t)}^i)}{(\sum_{j=1}^{N}{\rho_il_{a(t)}^j})^2}>0.\nonumber
\end{eqnarray}
\setlength{\arraycolsep}{5pt}
 \end{proof}

\subsection{Payment-Disaggregation Model}
\label{sec:paydisag}
In real-time market, the total payment paid by the PCP is a function of the effective day-ahead load, calculated in the previous section, total real-time load and market prices. In PJM market,  this payment is calculated using the following formulation~\cite{PJMM28}:

\begin{equation}
P_{(t)}=L_{e(t)}*p_{d(t)}+(L_{r(t)}-L_{e(t)})*p_{r(t)}
\label{TotalPayment}
\end{equation}
where \(L_{r(t)}=\sum_{i}{l_{r(t)}^i}\) is the aggregate real-time load consumption as a sum of individual consumptions, $p_{d(t)}$ and $p_{r(t)}$ are respectively day-ahead and real-time prices\footnote{Without loss of generality, we assume that all the prices are positive.} and $L_{e(t)}$ is the total effective day-ahead load as calculated by the PCP in the previous section. We note that in the PJM formulation there is no penalty for real-time load imbalance, but this doesn't change the generality of the formulation.

Once the total payment is known, the PCP has to disaggregate it for the share of each participating consumer. Since we assume that the PCP is a non-for-profit scheme, we expect the sum of these payments to be equal to the total payment made by the PCP to the market. It is also fair to assume that individual payments must be increasing functions of individual consumptions. Moreover, we expect the payment disaggregation scheme to be aligned with the overall incentive goals of the \textit{cooperative}, that is encouraging participants to either minimize their deviations or to time-shift their loads to lower total daily payments made by the PCP. To create this incentive structure, PCP divides consumers into two groups: The first group (\textit{deviation reducers}) are those who have helped with lowering the total deviation, to whom the PCP offers a potential reward by giving them a price choice between day-ahead and real-time prices. This incentive is in general costly for the PCP and its cost is assumed to be shouldered by the second group (\textit{deviation contributors}) who have contributed to the total deviation. These constraints and incentive mechanisms are captured more clearly in the following axioms:

\subsubsection{Balanced Payment}: In a non-profit \textit{cooperative}, sum of all payments made by consumers must be equal to what the PCP pays to the market, i.e., $\sum_{i=1}^{N}{P_{(t)}^i}=P_{(t)}$.

\subsubsection{Positive Marginal Price}: As long as the direction of imbalance remains unchanged, i.e., the sign of $\Delta_{(t)}=L_{r(t)}-L_{e(t)}$ remains unchanged, each consumer's payment is an increasing function of her load. i.e. ${\partial (P_t^i)}/{\partial (l_{r(t)}^i)} > 0$.

\subsubsection{Deviation Incentive}: The first group of consumers, those who have not contributed to the aggregate load deviation, will be given a choice to pick the price ($p_{r(t)}$ \textit{vs} $p_{d(t)}$) by which their deviation payment is calculated. This deviation payment can in general be positive or negative.

In order to formulate deviation incentives, as well as deviation responsibility in \textit{Axiom 4}, we define four deviation cases for each individual as follows: $ \pmb{case_1}:\ \Delta>0\ and\ \delta^i>0, \pmb{case_2}:\ \Delta<0\ and\ \delta^i<0, \pmb{case_3}:\ \Delta>0\ and\ \delta^i<0, \pmb{case_4}:\ \Delta<0\ and\ \delta^i>0$. Where $\Delta$ was defined earlier and ${\delta_{(t)}^i}=l_{r(t)}^i-l_{e(t)}^i$ for each consumer. Cases 3 and 4 represent \textit{deviation reducers} while cases 1 and 2 represent \textit{deviation contributors}.

Assuming all consumers are rational, and assuming individual payments are calculated using similar formulation as in (\ref{TotalPayment}), this axiom results in the following payments for \textit{deviation reducer} consumers:
\setlength{\arraycolsep}{0.0em}
\begin{eqnarray}
 P^i = \left\{ 
  \begin{array}{l l}
    l_e^i*p_{d}+(l_r^i-l_e^i)*\max\{p_d,p_r\}&\text{for $case_3$}\\
    l_e^i*p_{d}+(l_r^i-l_e^i)*\min\{p_d,p_r\}&\text{for $case_4$}
  \end{array} \right.
  \label{form:choicepayment}
\end{eqnarray}
\setlength{\arraycolsep}{5pt}

\subsubsection{Deviation Responsibility}: This axiom is complimentary to the previous axiom, and states that \textit{deviation contributors} shall be responsible for potential additional costs, posed on PCP, because of the incentives given to \textit{deviation reducers} as stated under \textit{Axiom B.3}. 

We now need to formulate individual payments that satisfy these four axioms. As the first step, we distinguish between scenarios depending on the direction of price change from day-ahead to real-time. We consider subscripted letters $a$ and $b$ for $p_r\geq p_d$ and $p_d<p_r$, respectively. For instance $case_{1a}$ combines the three conditions of $\Delta L_{t}>0\ ,\ \Delta L_{t}^i>0$, and $p_r\geq p_d$. 
  
When $\Delta>0$ and $p_{r}>p_{d}$ ($case_{1a},case_{3a}$) or  $\Delta<0$ and $p_{r}<p_{d}$ ($case_{2b},case_{4b}$), incentives for \textit{deviation reducers}, as stated in \textit{Axiom B.3}, does not add additional cost to the \textit{cooperative} (price choice is the real-time price in these cases), thus the deviation payment of all consumers can be calculated using $p_r$ and the individual payment for each consumer, regardless of the case they belong to, is calculated as follows: 

\begin{equation}
 P^i =l_e^i*p_{d}+(l_r^i-l_e^i)*p_r 
\label{eq:regularpayment}
\end{equation}

For such cases, the above equation satisfies all four axioms by definition. However, when $\Delta>0$ and $p_{r} < p_{d}$ ($case_{1b},case_{3b}$) or $\Delta<0$ and $p_{r}>p_{d}$ ($case_{2a},case_{4a}$), payment calculations need to be formulated differently to satisfy \textit{Axioms B1-B4} since \textit{deviation reducers} use day-ahead prices in the choice they are given as stated in \textit{Axiom B.3}, which in turn result in additional cost for the PCP.  This cost, as mentioned earlier, should be paid by \textit{deviation contributors}.

Similar to the effective day-ahead load model, the formulation for individual payments here is not unique. We construct individual payments as follows and will later show that it satisfies all the axioms:
\setlength{\arraycolsep}{0.0em}
\begin{eqnarray}
P_{(t)}^i=l_{r(t)}^i*p_{d(t)}+{\Delta} *(p_{r(t)}-p_{d(t)})*\frac{\delta^i}{\sum_{j \in S} {\delta^j }}
\label{form:costsharing}
\end{eqnarray}
\setlength{\arraycolsep}{5pt}
where \(S\) indicates the set of all \textit{deviation contributors} and an additional cost is added to their payment, proportional to their share of total deviation. 

\begin{proposition}
The payment model satisfies Axioms B.1- B.4. 
\end{proposition}
\begin{proof}
First, it's easy to show that the formulation is \textit{balanced} (Axiom B.1) and the sum of individual payments is equal to (\ref{TotalPayment}). In the cases when (\ref{eq:regularpayment}) applies, payments naturally add up to (\ref{TotalPayment}). However, for cases where (\ref{form:costsharing}) applies, we have: $P=\sum_{i\in N}{l_{r}^i*p_{d}}+\Delta*(p_{r}-p_{d})*{\sum_{i \in S}{\delta^i}}/{\sum_{j \in S} {\delta}^j}=(l_{d}+\Delta)*p_{d}+\Delta*(p_{r}-p_{d})=L_{d}*p_{d}+\Delta*p_{r}$.

Second, the proposed formulation is increasing with respect to the individual load, given the total balance and price direction (Axiom B.2). For cases where (\ref{eq:regularpayment}) applies, this is satisfied by definition. For those cases who use (\ref{form:costsharing}), for the sake of simplicity we put: $$f=\frac{\delta^i*\sum_{j=1}^N \delta^j}{\sum_{j \in S} {\delta^j}}=\frac{(x-c)(x+\alpha)}{(x+\beta)}$$
when \(x=l_{r}^i\), \(c=l_e^i\), \(\alpha<\beta\) and \(x>c\) (since the denominator is the sum of only positive deltas). Then, we prove that the first derivative of the payment function is strictly positive: 
\setlength{\arraycolsep}{0.0em}
\begin{eqnarray}
df/dx=(x^2+{2\beta}x+\alpha\beta - c\beta+c\alpha)/(x^2+{2\beta}x+{\beta}^2)\nonumber
\end{eqnarray}
\setlength{\arraycolsep}{5pt}
since \(\alpha\beta - c\beta+c\alpha<\alpha\beta - c\beta+c\beta=\alpha\beta<{\beta}^2\)
then \(df/dx<1\) and since \(dP_{(t)}^i/dx=p_{D(t)}+(p_{R(t)}-p_{D(t)})*df/dx\) and \(p_{R(t)}\geq0\)
then \(dP_{(t)}^i/dx>0\). 
Similar logic applies to the case with \(\delta^i<0\).

Finally, (\ref{form:costsharing}) ensures \textit{deviation responsibility} for $case_1$ and $case_2$ and (\ref{form:choicepayment}) ensures the consumer's choice for $case_3$ and $case_4$ (Axioms B. 3 and B.4). 
\end{proof}

\subsection{Incentive Compatibility}
In this section, we first show that, given the suggested scheme, participants behave truthfully in reporting their expected consumptions, thus the PCP can trust consumers. Then, and as a key result of this paper, we show that consumers are in general better off if they join a PCP, compared to a hypothetical scenario in which they directly participate in the market.
\begin{proposition}
In the proposed scheme, if a consumer assumes that the aggregate load is unbiased, regardless of the distribution of confidence for other consumers, being truthful about the expected load is a Nash equilibrium. 
\label{Proposition01}
\end{proposition}
\begin{proof}
We first prove a lemma that states that the effective bid, equal to expected real-time demand, minimizes the expected price; then show that truthful load provided by each consumer ensures the optimal expected bid (equal to expected demand).
\begin{lemma*}
Assuming normal distribution for real-time prices, with expected value equal to the day-ahead price ($E[p_r]=p_d$)\footnote{\ This assumption is in accordance with PJM market price distribution.}, if the total day-ahead load estimation is unbiased ($E[\Delta]=0$) accurate effective bid in the day-ahead market ($l_e^i=l_r^i$) results in a minimum expected electricity price for the consumer (see Appendix~\ref{Lemma01proof}).
\label{lemma01}
\end{lemma*}
From the above lemma, the accurate effective bid (equal to the expected real-time load) yields the least expected price. Accordingly, we set the day-ahead load in (\ref{form:risksharing}) equal to $l_r$: 
\setlength{\arraycolsep}{0.0em}
\begin{eqnarray}
l_e^i-l_r^i=l_a^i-l_r^i+\frac{\rho_i^2l_a^i(L_f-\sum_{j\in N}{l_a^j})}{\sum_{j\in N}{\rho_jl_a^j}}
\end{eqnarray}
\setlength{\arraycolsep}{5pt}

and for the expected value we have: 
\setlength{\arraycolsep}{0.0em}
\begin{eqnarray}
E&\left[l_a^i-l_r^i+\frac{\rho_i^2l_a^i(L_f-\sum_{j\in N}{l_a^j})}{\sum_{j\in N}{\rho_jl_a^j}} \right]\nonumber \\ 
&=l_a^i-E[l_r]+E\left[\frac{\rho_i^2l_a^i(L_f-\sum_{j\in N}{l_a^j})}{\sum_{j\in N}{\rho_jl_a^j}} \right]\nonumber
\end{eqnarray}
\setlength{\arraycolsep}{5pt}

In the last term of the above equation, lets assume that we can define positive upper and lower boundary for the denominator corresponding to $l_a^i$: $B_1<\sum_{j\in N}{\rho_jl_a^j}<B_2$. For the lower boundary, if $l_a^i \to E[l_r]^-$,  we have: 
\setlength{\arraycolsep}{0.0em}
\begin{eqnarray}
E[l_e^i-l_r^i]&>f_{L}(l_a^i)=l_a^i-E[l_r]+\frac{E[\rho_i^2 l_a^i (L_f-\sum_{j\in N}{l_a^i})]}{B_2}\nonumber\\
&=l_a^i-E[l_r]+\frac{\rho_i^2 l_a^i (E[L_f-\sum_{j\in N \setminus i}{l_a^j}]-l_a^i)}{B_2}\nonumber
\end{eqnarray}
\setlength{\arraycolsep}{5pt}
because other consumers and the PCP are assumed to be truthful and unbiased, we have:
$$
E[L_f-\sum_{j\in N\setminus i}{l_a^j}]=E[l_r]
$$
if we set $l_a^i=E[l_r]$, the RHS of the inequality approaches zero: $\lim_{l_a^i \to E[l_r]^-} f_{L}(l_a^i)=0$.
 In addition, $B_1$ results in an upper boundary:
\begin{equation}
f_{U}(l_a^i)=l_a^i-E[l_r]+\frac{\rho_i^2 l_a^i (E[l_r]-l_a^i)}{B_1}\nonumber
\end{equation} 
where $\lim_{l_a^i \to E[l_r]^-} f_{U}(l_a^i)=0$. Since both upper and lower boundary approach zero for $l_a^i \to E[l_r]^-$, using a similar logic these boundary approach zero for $l_a^i \to E[l_r]^+$. We can then conclude: 
$$
\lim_{l_a^i \to E[l_r]} E[l_e^i-l_r^i]=0
$$
which indicates that condition to other consumers being truthful, no single consumer has incentive to change her truthful strategy, thus being truthful is a Nash equilibrium. 
\end{proof}

\begin{proposition}
\label{Proposition04}
For a given consumer, if PCP is unbiased with respect to aggregate load estimation, relying on PCP's effective bid is a (weakly) dominant strategy.
(see Appendix~\ref{Proposition02proof}) 
\end{proposition}


%



\section{Short Term Load Forecasting}
\label{sec:forecasting}

The PCP uses a variation of double seasonal autoregressive exponential smoothing (DSAES) for forecasting aggregated load~\cite{taylor2003short,taylor2010exponentially}. 
The original formulation is taken from~\cite{taylor2007short} as: 
 \begin{equation}
\hat{L}_f(t+k)=b_t+d_{t-s_1-k_1}+w_{t-s_2+k_2}+\phi^ke_t
\label{form:forecast}
 \end{equation}

 where $\hat{L}_f(t+k)$ is the load forecast for the next $k$ steps of current time interval \(t\), \(s_1\) and \(s_2\) are the number of time intervals in a day and week (24 and 168 respectively), \(b_t\) is the smoothed level, \(d_t\) and \(w_t\) are the seasonal index for daily and weekly cycles, \(k_1=[(k-1)\ mod\ s_1]+1\) and \(k_2=[(k-1)\ mod\ s_2]+1\), and \(k\) is forecasting lead time. The term involving \(\phi\) is the autoregressive adjustment for first-order residual autocorrelation. 

 \begin{wrapfigure}[18]{l}{0.52\columnwidth}
 \vspace{-5pt}
 \centering
    \includegraphics[width=0.55\columnwidth]{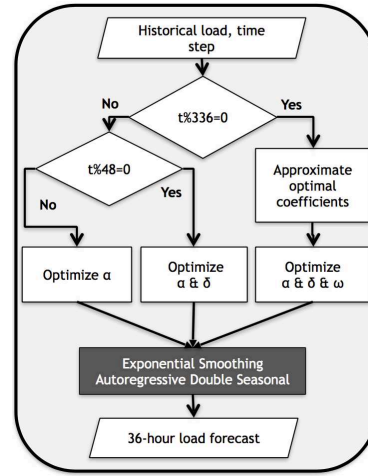}
   \vspace{-20pt}
  \caption{Dynamic load forecasting model (opt: optimize)}
 \label{fig:flowchart}
\end{wrapfigure}

A critical step in this model is to initialize the three main model parameters for daily, weekly and autoregressive adjustments. The initialization method introduced in~\cite{taylor2010exponentially} uses \(10^5\) random initial points for each model parameter as well as the Newton method for root estimation. 

This model however has two practical problems for autonomous systems: First, the level of sensitivity of the prediction accuracy to these parameters makes static parameter estimation less appealing. Moreover, the initialization algorithm is rather computationally intensive and demands disproportionate computational resources. Instead, we modified the model using a dynamic algorithm that calculates forecasting model parameters once a day to predict the 36 hours of load time series. Note that the day-ahead bid includes all the expected values of 24 consecutive hours of next day load, which is in accordance with the day-ahead market functionality. Then, using 12- to 36-hour load forecasting, PCP can submit the day-ahead bid at 6PM of the day before the transaction\footnote{The complexity of the second approach is reduced from $O(T*10^{15}*F(n))$ (for each data point) to $O(T*10^{10}*F(n)/7)$ when T is the number of historical time steps used (672 for 4 weeks) and F(n) is the complexity of Newton method for n refers to digit precision. The computational intensiveness is however not a major problem in the proposed PCP as it is calculated only once a day.}(see Fig.~\ref{fig:flowchart}).

This model was implemented in Python on a sample consumers load data~\cite{laodDataSet}. The results show 9\%-15\% reduction in forecasting error, an improvement that is the result of using dynamic parameter updating, in comparison with the more conventional static parameter methods (Table~\ref{tab:MAPE}).

\begin{table}[!t]
 \renewcommand{\arraystretch}{1.3}
 \caption{MAPE Measure of DSAES Method with Fixed and Dynamic Parameters }
  \label{tab:MAPE}
\centering
\begin{tabular}{c| | l | l c}
\hline
 \bfseries Lead Hours    &  \bfseries Dynamic Parameters & \bfseries Fixed Parameters    \\
\hline
12   & 0.0474 & 0.0552      \\
\hline
24           &0.0424 & 0.0468   \\
\hline
36             & 0.0497 & 0.0541    \\
\hline

\end{tabular}

\end{table}

\section{Simulation}
In this section, we first provide some illustrative computational results for the axiomatic models developed in Section~\ref{sec:ax_model}. We the provide a simulated model of the proposed scheme, using real-world market data.

\subsection{Numerical Illustration}
Fig.~\ref{fig:paymentdisaggregation} shows the results for individual unit price as a function of individual load deviation as presented in Section~\ref{sec:paydisag}. Each figure illustrates the price sensitivity to load deviation for four scenarios of real-time price deviation (RPD) where ($\Delta$) represents the aggregate load deviation.  The relative price is decreasing as a function of real-time load when the aggregate load is negative, although the marginal price is increasing (see Fig.~\ref{fig:paymentdisaggregation}a). The opposite case holds for the positive aggregate load deviation (see Fig.~\ref{fig:paymentdisaggregation}d). In the case with zero aggregate deviation, we can see that the optimum price pertains to zero individual deviation (see Fig.~\ref{fig:paymentdisaggregation}b). Finally, if the individual deviation can offset the aggregate deviation, one can observe a discontinuous price because the consumer switches from being a \textit{deviation reducer} to a \textit{deviation contributor} (see Fig.~\ref{fig:paymentdisaggregation}c). 

In Fig.~\ref{fig:IndDev}, we show the results for the electricity unit price as a function of individual load deviation (day-ahead bid minus the expected real-time load). In this figure, the optimum level of deviation is very close to zero for different levels of MAPE. However, we can see a shift toward right for higher levels of estimation error. The reason is a higher price sensitivity to lower values in denominator, i.e., the real-time load. Fig.~\ref{fig:AggDev} shows the expected price versus the individual deviation when the expected total deviation ($E[\Delta]$) varies. The expected price is higher for higher deviation when there is negligible aggregate deviation. However, the expected price decreases if a consumer's load deviation is against the aggregate load deviation. For instance, a consumer benefits if its load deviation is positive while the aggregate deviation is negative and vice versa. This results are in agreement with the results shown in Fig.~\ref{fig:paymentdisaggregation}. This opportunity ensures that any consumer with knowledge of aggregate bias may modify her expected load to the PCP, which in turn reduces the aggregate deviation. This phenomena can be considered as a negative feedback loop between the expected aggregate bias and the individual bias. 

\begin{figure}[t]
\centering
 \includegraphics[width=0.48\textwidth]{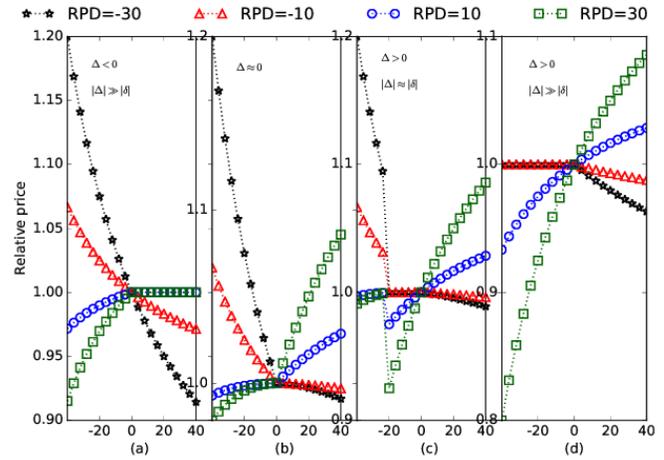}
\caption{Relative price deviation vs load deviation and \textit{real-time price deviation}(RPD): (a) aggregate balance is negative and $|\Delta|$ is significantly larger than $|\delta|$, (b) aggregate balance is close to zero, (c) aggregate balance is positive and $|\Delta|$ is comparable to $|\delta|$, (d) aggregate balance is positive and $|\Delta|$ is significantly larger than $|\delta|$, ($\Delta$:aggregate deviation, $\delta$: individual load deviation).}
\label{fig:paymentdisaggregation}
\end{figure}

\begin{figure}[!t]
    \centering
    \mbox{\subfloat[]{\label{fig:IndDev}\includegraphics[clip,width=0.88\columnwidth]{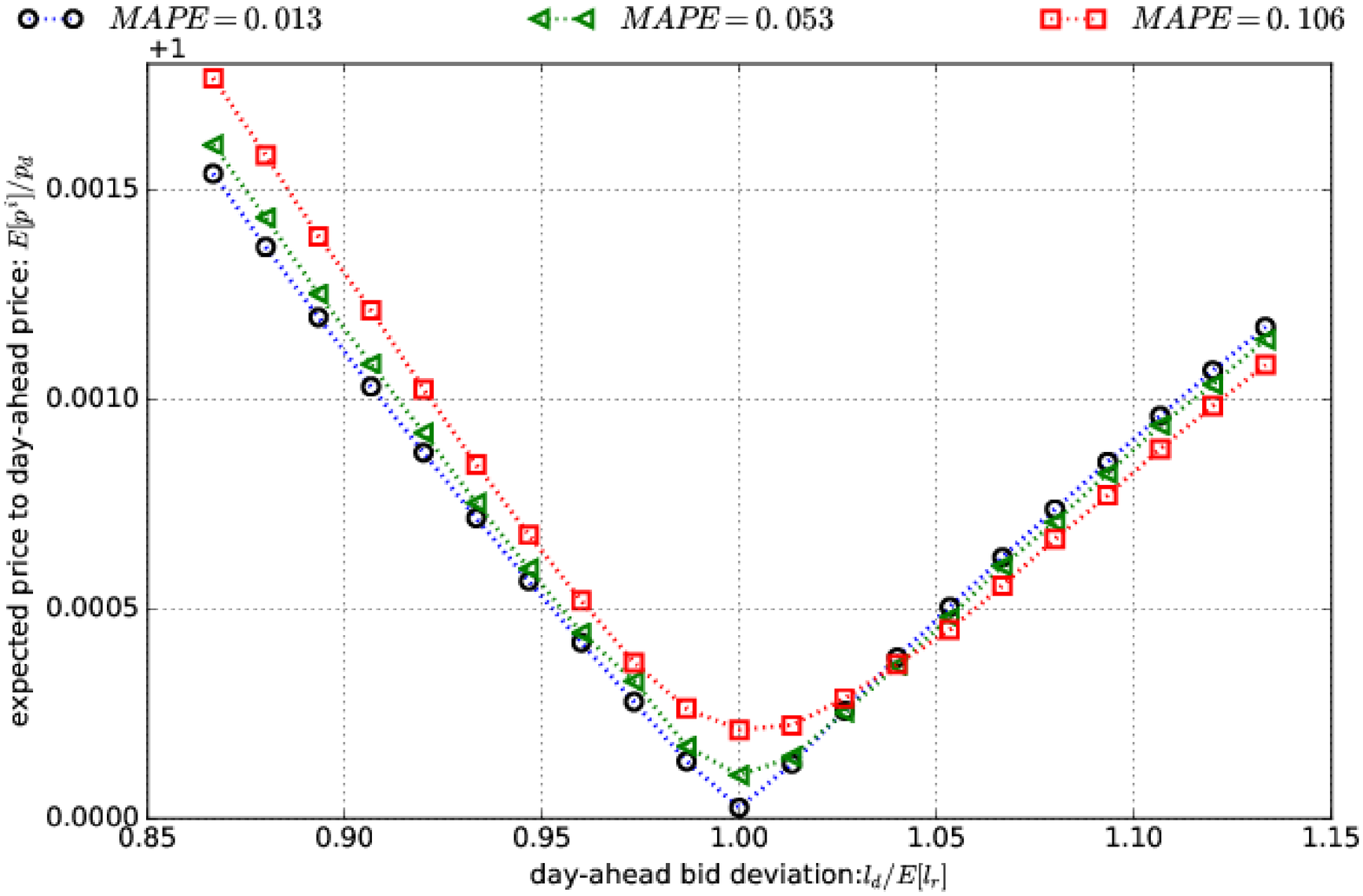}}}
      \mbox{\subfloat[]{\label{fig:AggDev} \includegraphics[clip,width=0.88\columnwidth]{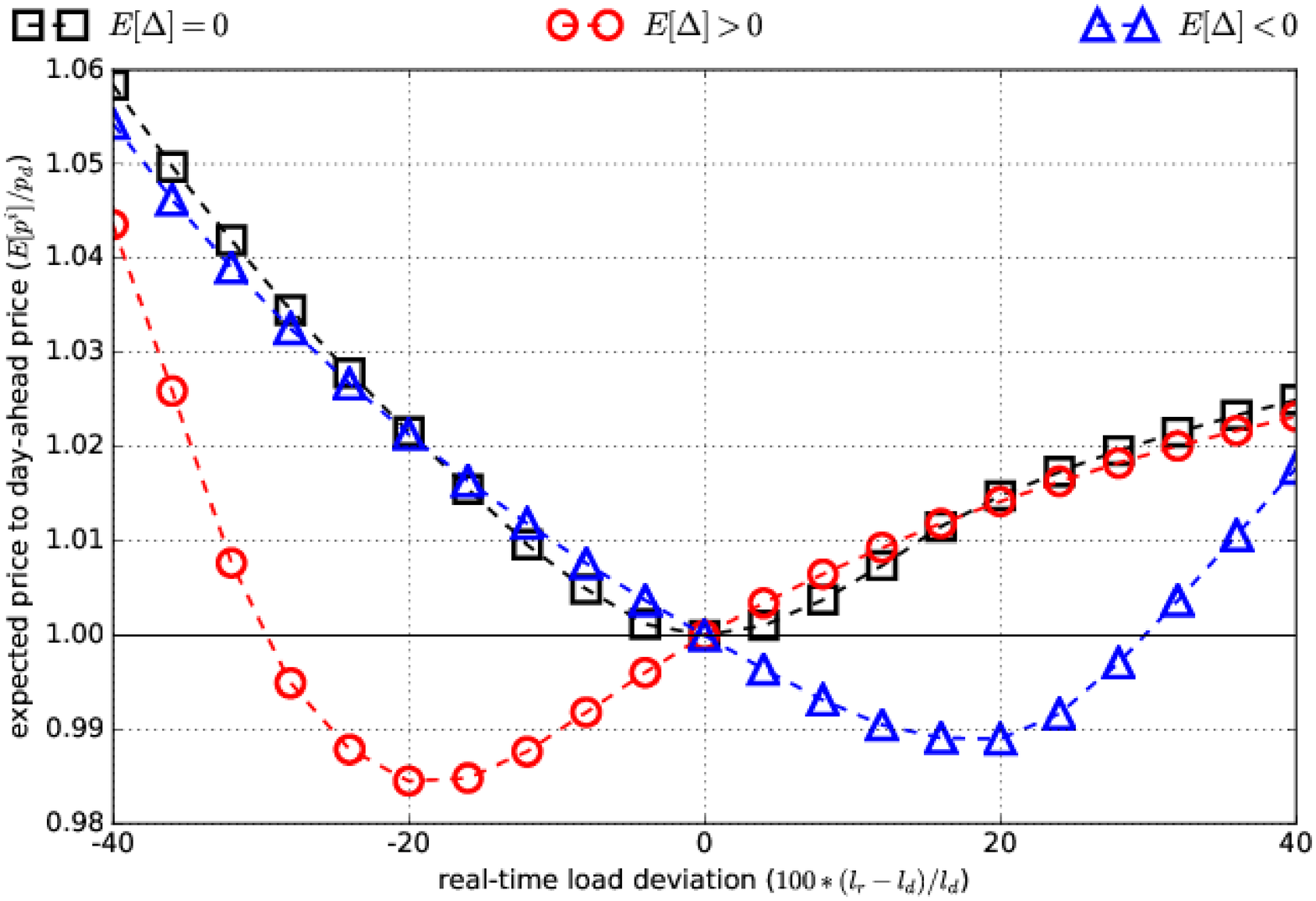}}}
    \caption{Electricity unit price vs the effective bid deviation and real-time load deviation: a) the expected relative price versus the bid deviation with unbiased \textit{cooperative}, b) the expected relative price vs real-time deviation with biased \textit{cooperative}.}
    \label{fig:expectedPrice}
\end{figure}

\subsection{Agent-Based Simulation}
\label{sec:agentbasedsim}
We study the performance of the PCP using real-world electricity day-ahead and real-time prices and the realistic load profiles in an agent-based model. We collect a sample repository of simulated  load profiles according to the PJM geographic area~\cite{laodDataSet}. To benchmark the proposed schemes, we compare the simulation results with \textit{real-time pricing} (RTP) scheme. In RTP, each consumer places a bid in the day-ahead market and the payment is calculated using her effective bid, real-time consumption, and market prices similar to (\ref{TotalPayment}).  Simulations, including all the models for the PCP and consumers are developed  using Object-Oriented Python.

The model inputs include the trading time frame ($T$), the number of consumers ($N$), hourly load profile for each consumer ($L_i$), market day-ahead prices ($P_{d}$), market real-time prices ($P_{r}$), and the number of simulation rounds ($M$). The estimated load and the confidence factors are calculated and announced by each consumer agent. The latter is assigned using a uniform distribution between 0 and 1 and is updated by each consumer based on the \textit{cooperative}'s historical advantage to her. The model output include the list of relative prices ($p_{pcp}/p_d$ or $p_{rtp}/p_d$) for each set of consumers with similar level of MAPE. 

The trading time frame is set to 3600 hourly time intervals for Feb-July 2015 and we consider the same set of consumers (one hundred) for PCP and RTP schemes. To represent a consumer's load estimation error, we use MAPE measure that evenly varies between 0.02 and 0.20. The MAPE of our aggregate forecasting model, presented in Section~\ref{sec:forecasting}, is less than 0.05 in this model. To diminish the effect of a single load profile (rather than MAPE) on the outcome, we consider 50-round simulations with random MAPE assignments to consumers. For each consumer, we assume that the confidence factor ($\rho_i$) is a random value between 0 and 1 in the first iteration, drawn from a uniform distribution. For the purpose of this simulation, we further assume that each confidence factor is updated to the average of its previous value and 1(0) in case the PCP average price is lower(higher) than RTP average price during the last 24 hours. Fig.~\ref{fig:avgconfidence} shows the evolution of confidence factor (average $\bar{\rho}$ for all consumers and $\rho_i$ for consumers with different MAPE levels). The average confidence factor converges to lower values for consumers with higher MAPE, mainly because they are more likely to be a \textit{deviation contributor}. The statistical measures for each data point are generated using a set of $1.8e^6$ calculated relative prices\footnote{Includes the 50 rounds of simulations on 3600 time steps.}. The statistical measures include relative price median, standard deviation and various percentiles of each set of sample points. To better represent the overall trend, in Fig.~\ref{fig:pricedistribution}, we don't show all of the data points. 

 The first simulation is aimed at determining the effect of load forecasting error on consumer's \textit{relative price} and \textit{relative price standard deviation}. Fig.~\ref{fig:pricedistribution} shows the results for unit price distribution. In all figures, the x-axis shows the MAPE and the price variation is expectedly increasing corresponding to the forecasting error. Fig.~\ref{fig:pricedist} compares the price percentiles between PCP and RTP. The median relative price is equal to one for all consumers because expected values of real-time and day-ahead prices are very close to each other in the PJM market. In PCP, the price distribution is skewed toward lower prices because, assuming an unbiased load forecasting and price distribution and positive deviation, the lower values of $p_r$ and $p_d$ statistically apply to 25\% of the cases.  Also, $p_r$ applies to other 50\% of cases and a finally quarter of the prices remain between the two values (50\% higher and 50\% lower than $p_d$). So, less than half of the relative prices can stay higher than the median price and this is why the distributions are skewed toward the lower prices. Note that same logic can be applied to negative deviation, however, this logic only applies to consumers whose deviation cannot offset the aggregate deviation. For higher levels of MAPE, the distribution approaches a normal distribution with mean value equal to one.

\begin{figure*}[!t]
    \centering
    \subfloat[]{\label{fig:pricedist} \includegraphics[width=2.2in]{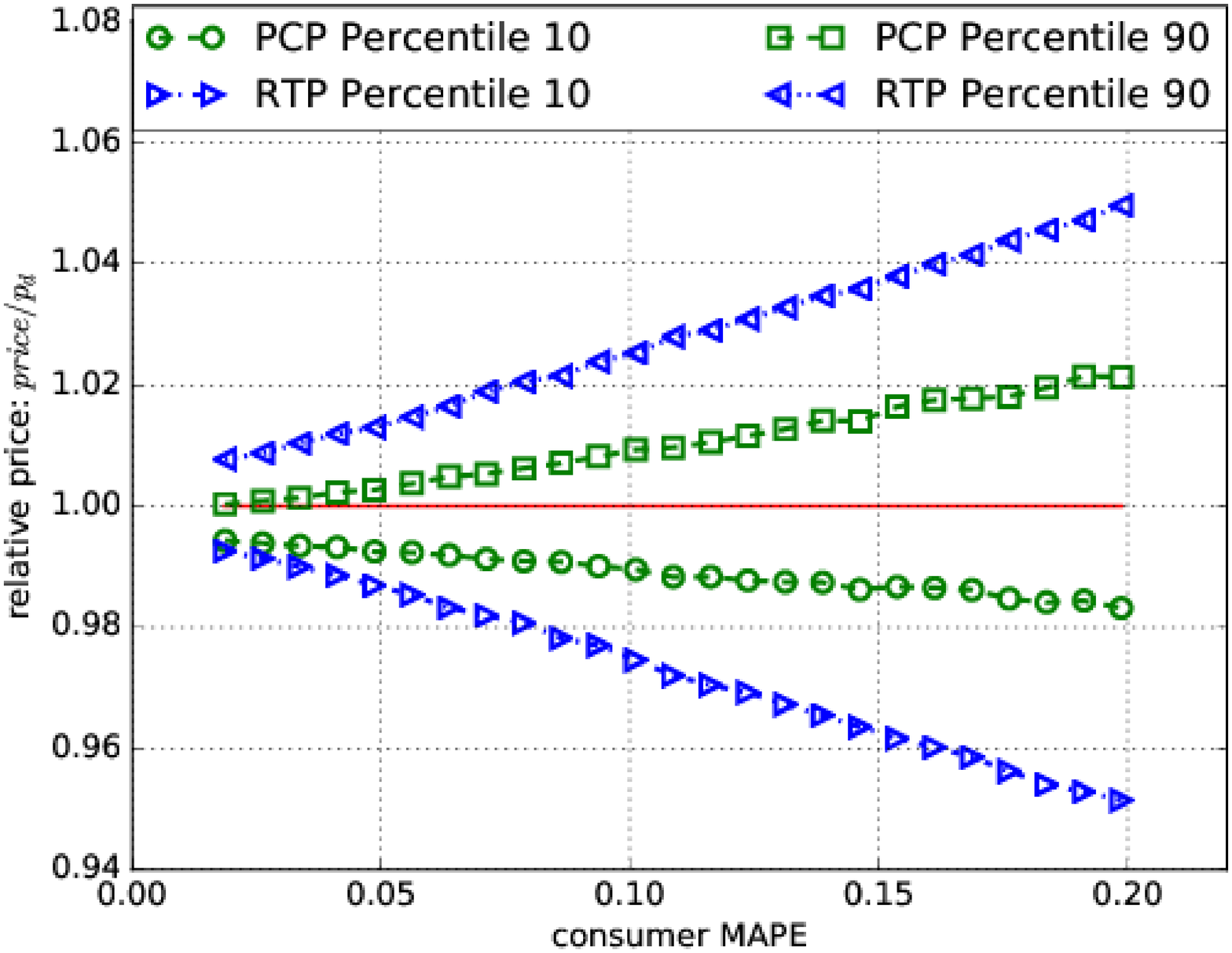}}
        \hfil
        \subfloat[]{\label{fig:std} \includegraphics[width=2.2in]{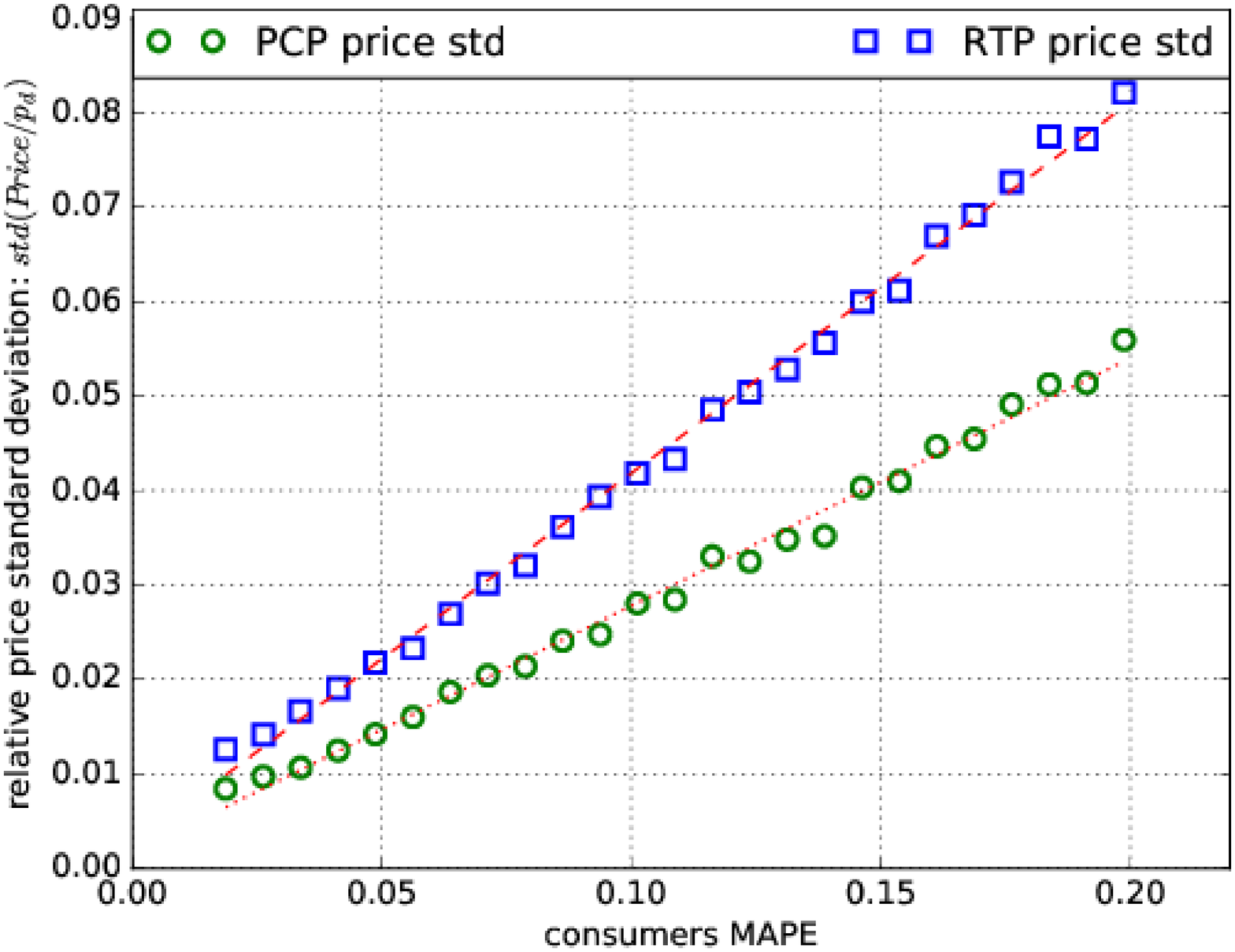}}
        \subfloat[]{\label{fig:avgconfidence} \includegraphics[width=2.2in]{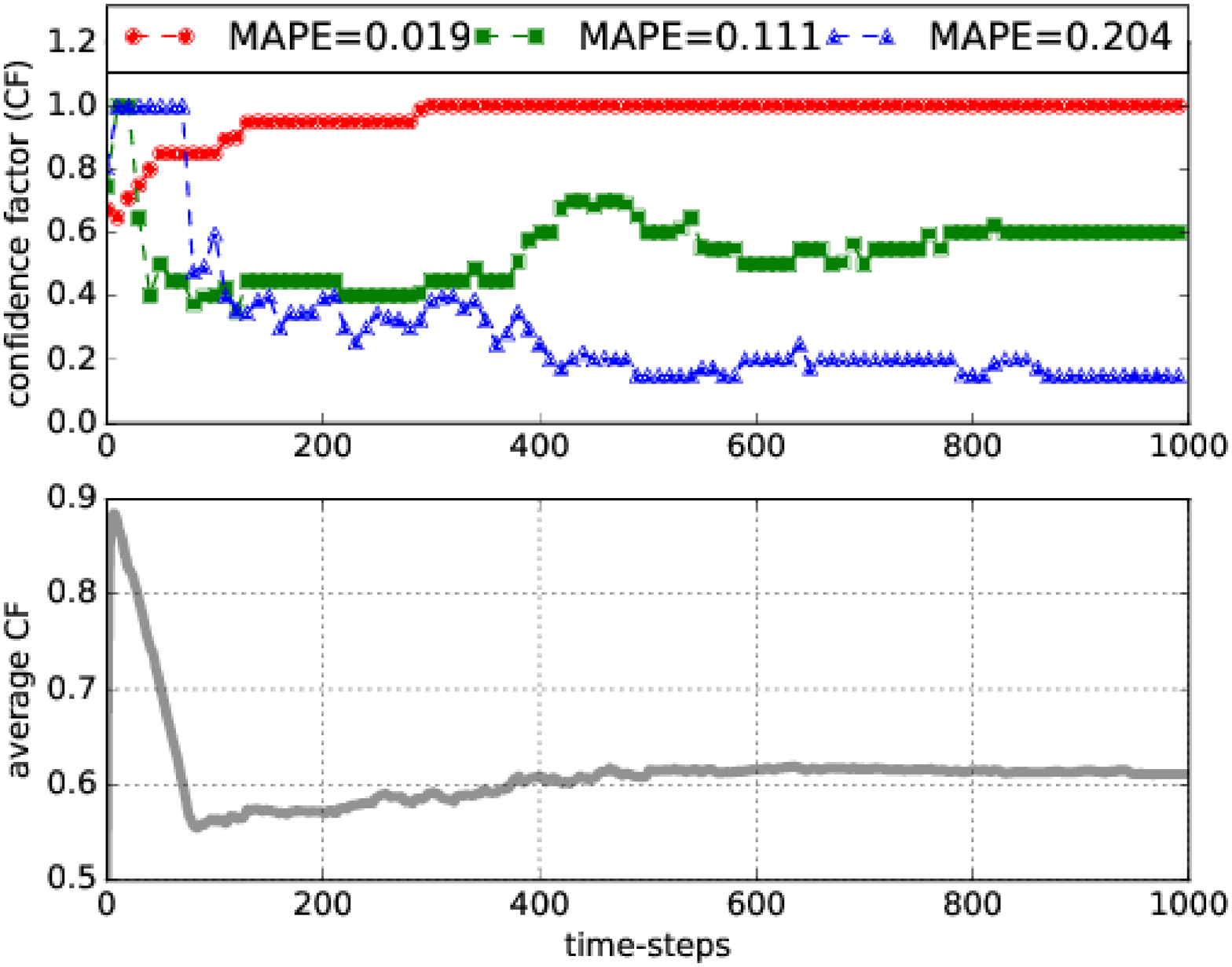}}

        \label{fig:standarddeviation}
    \caption{Simulation results for the electricity price as a function of the load estimation error (MAPE): (a) Relative unit price (the final price divided by the day-ahead price) percentiles  and (b) Relative price standard deviation (normalized standard deviation with the day-ahead price) for RTP and PCP. The agent-based simulation includes 100 consumers, (c) confidence factor evolution in multi-agent simulation.}
    \label{fig:pricedistribution}
\end{figure*}

\section{Conclusion and Discussion}
In electricity wholesale markets, day-ahead prices offer clear incentives to consumers for flattening their load-profiles. However, risks associated with price volatilities in real-time market, among other things, discourage consumers from active participation in the electricity market. Any autonomous scheme that offers consumers lower risk associated with real-time price volatility can incentivize active demand-side management based on market prices. In this paper, we introduced an autonomous \textit{cooperative} agent that applies load forecasting along with axiomatic day-ahead bid and payment sharing model for the electricity wholesale markets. We contributed to (i) axiomatic models for day-ahead bid calculation and cost sharing and (ii) a dynamic coefficient updating in a competitive forecasting technique. 

The paper offers a novel incentive-compatible scheme for electricity \textit{cooperative}s to achieve a collective goal of reducing price risk for the participants while enabling the prospect of reduced average price for individual consumers, using a bottom-up, axiomatic approach. It is mathematically proven that the provided scheme is not only beneficial for individual consumers to participate in, but also provides incentives to participants to adopt a truthful behavior when interacting with the scheme. Using agent-based simulation, we tested the proposed scheme with real-world data of PJM markets and realistic consumer load profiles. The results demonstrate that the PCP scheme, compared to a hypothetical RTP scheme -in which individual consumers interact directly with the market - reduces electricity price variations for all consumers. In addition to being truthful to PCP about their estimated load, consumers are incentivized to decrease their estimation variance since higher estimation accuracy leads to higher confidence factor and lower prices.

\appendices
\section{Proof of the Lemma~\ref{lemma01}}
\label{Lemma01proof}
Let's assume the following probabilities for a specific time interval $t$: $\rho=P[p_r>p_d]$, $\omega=P[L_r>L_e]$, and $\theta=P[l_r^i>l_e^i]$ for the consumer $i$. According to the \textit{partition theory}, the expected payment by the consumer is (for simplicity, we avoid repetitive $t$ and show \textit{case} with letter \textit{c}): 

\setlength{\arraycolsep}{0.0em}
\begin{eqnarray}
&E[\frac{P_i}{l_r}]=\rho \omega \theta E[\frac{l_ep_d+\delta^i p_r}{l_r}|c_{1a}]\nonumber\\
&{+}\:\rho \omega (1-\theta) E[\frac{l_ep_d+\delta^i p_r}{l_r}|c_{2a}]+\rho(1-\omega) \theta E[p_d|c_{3a}]\nonumber\\
&{+}\:\rho(1-\omega)(1-\theta)E[p_d+\frac{\delta^i \Delta (p_r-p_d)}{\Delta^{S^-} l_r}|c_{4a}]\nonumber\\
&{+}\:(1-\rho) \omega \theta E[p_d+\frac{\delta^i \Delta (p_r-p_d)}{\Delta^{S^+} l_r}|c_{1b}]\nonumber\\
&{+}\:(1-\rho) \omega (1-\theta) E[p_d|c_{2b}]\nonumber\\
&{+}\:(1-\rho) (1-\omega) \theta E[\frac{l_ep_d+\delta^i p_r}{l_r}|c_{3b}]\nonumber\\
&{+}\:(1-\rho) (1-\omega) (1-\theta) E[\frac{l_ep_d+\delta^i p_r}{l_r}|c_{4b}]
\label{eq:expectedprice01}
\end{eqnarray}
\setlength{\arraycolsep}{5pt}

This can be combined to much simpler expression. As the result, the equation equals to:
\setlength{\arraycolsep}{0.0em}
\begin{eqnarray}
&E[\frac{P_i}{l_r}]=Kp_d+\rho \omega E[\frac{l_e}{l_r}]E[p_d-p_r|p_r>p_d]\nonumber\\
&\phantom{=}+(1-\rho) (1-\omega) E[\frac{l_e}{l_r}]E[p_d-p_r|p_r<p_d]\nonumber\\
&\phantom{=}+\rho(1-\omega)(1-\theta)E[\frac{\delta^i \Delta}{\Delta^{S^-} l_r}|case_1]E[p_r-p_d|p_r<p_d]\nonumber\\
&\phantom{=}+(1-\rho) \omega \theta E[\frac{\delta^i \Delta}{\Delta^{S^+} l_r}|case_4]E[p_r-p_d|p_r>p_d]\nonumber
\label{eq:expectedprice02}
\end{eqnarray}
\setlength{\arraycolsep}{5pt}

In the last expression, $K$ is a constant, the $p_r$ has normal distribution with the mean $p_d$, then $\rho=0.5$. Also, the expected value of $\Delta$ is assumed to be zero by consumers, so $\omega=0.5$. Also we have: 
\begin{equation}
E[p_d-p_r|p_r>p_d]=-E[p_d-p_r|p_r<p_d]\nonumber
\end{equation}
We simplify the final formulation accordingly as: 

\setlength{\arraycolsep}{0.0em}
\begin{eqnarray}
E&[\frac{P_i}{l_r}]=p_d+0.25 \alpha \theta E[\frac{\delta^i \Delta}{\Delta^{S^+} l_r}|l_r>l_e,\Delta>0] \nonumber\\  
&\phantom{=p_d+}-0.25 \alpha (1-\theta) E[\frac{\delta_i \Delta}{\Delta^{S^-} l_r}|l_r<l_e,\Delta<0]
\label{eq:expected04}
\end{eqnarray}
\setlength{\arraycolsep}{5pt}

when $\alpha=E[p_r-p_d|p_r>p_d]>0$. Also by definition $E[\frac{\delta^i \Delta}{\Delta^{S^+} l_r}|l_r>l_e,\Delta>0]$ is a positive number ($\Delta^{S^+}>0$) and $E[\frac{\delta^i}{l_r}|l_r<l_e]$ is negative ($\Delta^{S^-}<0$); this ensures that $
E[\frac{P_i}{l_r}]\geq p_d$ and day-ahead bid equal to real-time consumption ensures the lowest expected price equal to $p_d$.

\section{Proof of the Proposition~\ref{Proposition04}}
\label{Proposition02proof}
In the cases when other consumers are biased in their estimation while the PCP is not biased, according to (\ref{form:risksharing}) the positive estimation bias by other consumers will reduce the expected value of day-ahead load by consumer $i$, also a negative estimation bias will increase the expected value of day-ahead bid. In other words, $E[\Delta]>0$ leads to $E[\delta]<0$ and vice versa if the consumer has been truthful with its estimation. 

Also, based on (\ref{eq:expected04}) the final expected price is combination of first part, with expected value of $p_d$ and the second part, including positive values. Lets assume that the $\omega=P[\Delta>0]>0.5$, then the most influential term in the equation will be: 
$$
\omega \theta E[\frac{\delta^i \Delta}{\Delta^{S^+} l_r}|l_r>l_e,\Delta>0]
$$
and since $\omega$ is already increased, we can reduce the coefficient $\omega*\theta$ by reducing the probability of positive load deviation $P[\delta^i>0]=\theta$. Since this effect holds for any $E[\Delta]>0$ with a symmetric distribution, any individual deviation that balances the aggregate deviation can reduce the expected price. The same logic holds for cases with $E[\Delta]<0$ or $\omega=P[\Delta>0]<0.5$ when reducing the aggregate deviation benefits the individual consumer. 
In simulations, this effect can further reduce the final price to $p<p_d$, as it is illustrated in Fig.~\ref{fig:AggDev}. 

\bibliographystyle{IEEEtran}
\bibliography{references}

\begin{IEEEbiography}[{\includegraphics[width=1in,height=1.25in,clip,keepaspectratio]{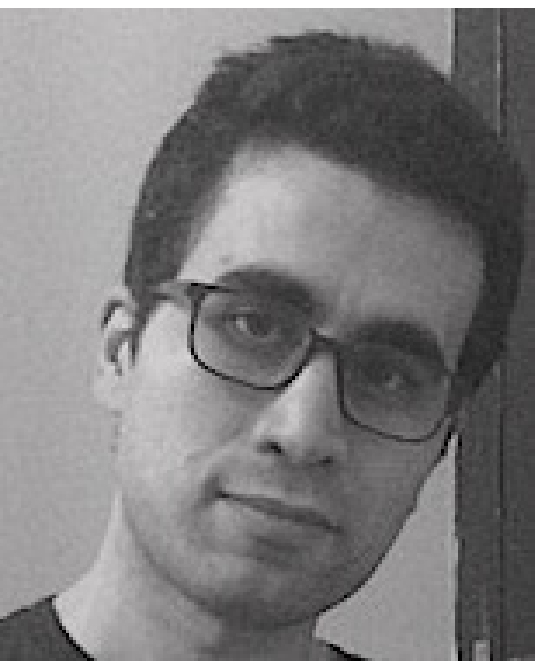}}]{Abbas Ehsanfar}
 is a Ph.D. student at the school of Systems and Enterprises Engineering at Stevens Institute of Technology. He completed a Bachelors degree in Electrical Engineering at Sharif University of Technology in Tehran, and his Master's degree at the same university. Abbas's research interests involve developing intelligent autonomous algorithms, \textit{cooperative} schemes, and learning frameworks applicable to adaptive systems in general and electricity smart grids in particular.
\end{IEEEbiography}

\begin{IEEEbiography}[{\includegraphics[width=1in,height=1.25in,clip,keepaspectratio]{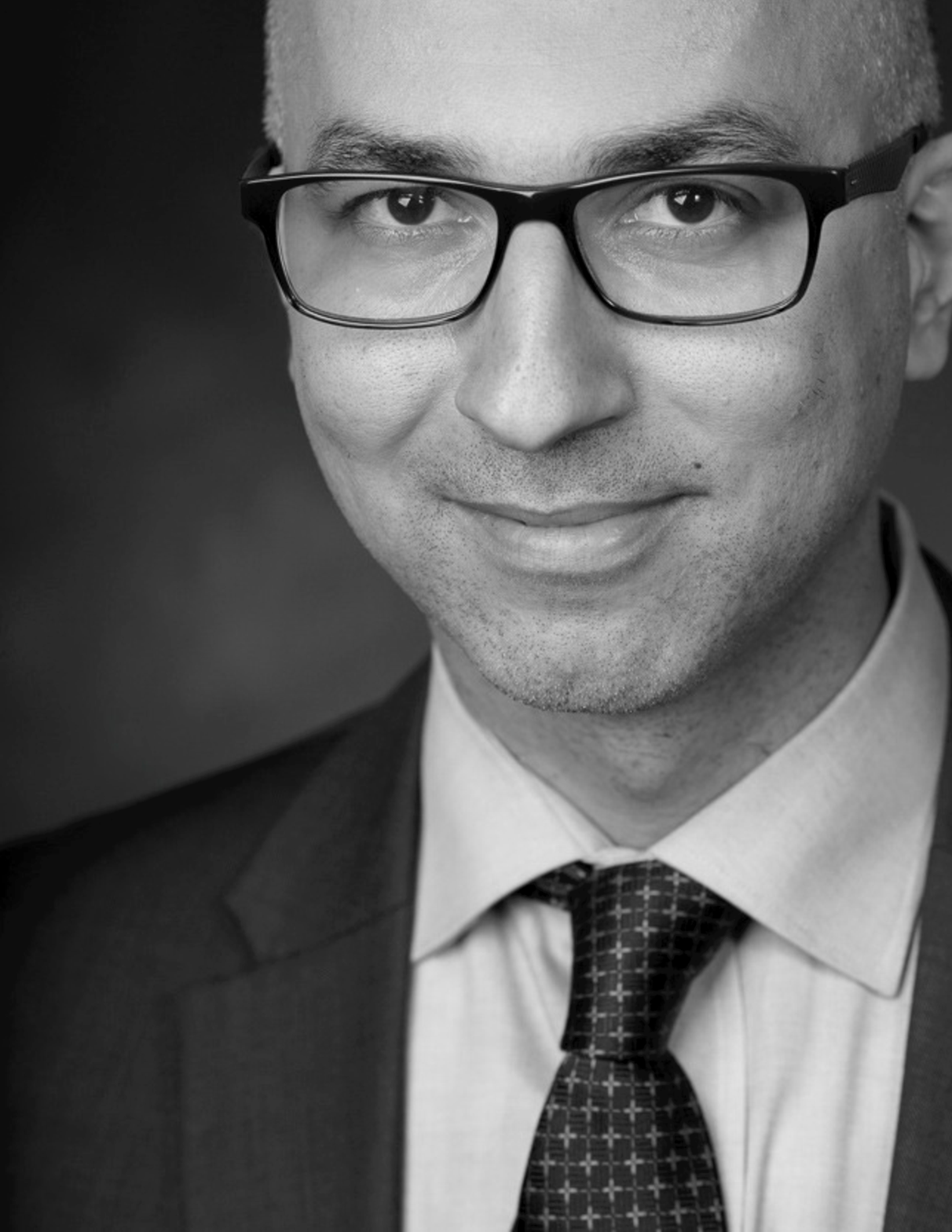}}]{Babak Heydari}
received the B.Sc. degree from Sharif University of Technology in Tehran, Iran, in 2002 and the Master's and Ph.D. degrees in electrical engineering from University of California at Berkeley in 2005 and 2008 respectively. He also received a graduate certificate in management of technology and economics at UC Berkeley in 2007. 

Following his graduation, he worked as entrepreneur in Silicon Valley start-ups until he joined Stevens Institute of Technology in 2011 as assistant professor where he is currently a faculty at the School of Systems and Enterprises and the director of Complex Evolving Networked Systems Lab. Prof. Heydari has a diverse set of research interests and does interdisciplinary research at the intersection of engineering, economics and systems sciences. His research topics include emergence and evolution of collective behavior in social and socio-technical networks, architecture and resource-sharing in socio-technical systems and data-driven policy analysis. His research has been funded by NSF, DARPA, INCOSE, SERC, and a number of private corporations. He was the technical chair of the 4th International Engineering Systems Symposium, CESUN2014 and is an associate editor of the Wiley journal of systems engineering. Prof. Heydari is the recipient of National Science Foundation CAREER Award in 2016.
\end{IEEEbiography}

\end{document}